\documentclass[pra,jmp,amsmath,amssymb,12pt,
reprint,showkeys,showpacs]{revtex4-1}

\usepackage{amssymb}
\usepackage{dsfont}
\usepackage{mathrsfs}
\usepackage{amsthm}
\usepackage[bookmarks]{hyperref}
\usepackage{graphicx}

\usepackage{pstricks}
\usepackage{pst-3dplot}

\usepackage{amsmath}
\usepackage{amscd}
\usepackage{mathrsfs}
\usepackage{multirow}
\usepackage{dcolumn}

\hypersetup{
            breaklinks,
            colorlinks,
            urlcolor=[rgb]{0,0,1},
            linkcolor=[rgb]{0,0,1},
            citecolor=[rgb]{0,0,1},
}

\newtheorem{definition}{Definition}[section]
\newtheorem{remark}{Remark}[section]
\newtheorem{theorem}{Theorem}[section]
\newtheorem{lemma}{Lemma}[section]

\begin{document}

\title{A Necessary and Sufficient Condition for Entanglement Witness}
\author{F. Masillo}
\email{masillo@le.infn.it}
\affiliation{Dipartimento di Fisica,  Universit\`{a} del Salento, I-73100 Lecce, Italy}

\date{\today}
\keywords{Quantum Entanglement, Positive Maps, Entanglement Witness}

\begin{abstract}
An Entanglement Witness (EW) is an observable that permits the identification of Entangled States.
While the structure of separable (non Entangled) states is very complicated and not yet completely understood,   the qubit, \emph{i.e.} the simplest non trivial quantum system, has a satisfactory geometrical  description in terms of Bloch spheres. In this paper, we will show that this geometrical representation of  qubits  can be used in understanding  separable states and, in particular, in the formulation of  a simple necessary and sufficient condition for EW.
\end{abstract}

\maketitle


\section*{Introduction}

Quantum  entanglement \cite{EPR35}, \cite{S35} is the responsible of the most fascinating quantum mechanical effects \cite{E91,BW92, BBCJP93, EJ96} and, in particular, it  plays a crucial role in all quantum informational processes that have a significant speed up with respect to the analogous classical ones \cite{NC00}.
It follows immediately that the identification and  manipulation of entangled states play a crucial role in this fascinating field of modern physics.

The search of valid criteria for the  identification of entanglement is a long history \cite{HHHH09} and simple operational criteria exist only  for low dimensional cases \cite{P96,HHH96}.

In this paper we concentrate our attention only on the set of \emph{entanglement witness} (EW) \cite{P96,HHH96,T00},
\emph{i.e.} on the set of linear functional, defined  on the set of states of a multipartite quantum systems, that are  
positive on the set of separable states.

The existence of EW  can be proven applying  the Hahn-Banach theorem to the convex set of separable states and the compact convex set containing only one entangled state \cite{HHH96}.
In particular, as an immediate consequence, it is possible to prove that  for every entangled state at least an EW  exists  that ``separates'' it from the set of separable states.
Moreover, noticing that EW are observables, the identification of entangled states is reduced to  measure  particular physical quantities.

Another important motivation to investigate EW is due to the famous Choi-Jamio{\l}kowski (CJ) isomorphism \cite{J74,C82}.
This last in fact, permits the construction of positive, but not completely positive (CP),  maps starting from the set of bipartite EW and vice versa \cite{HHH96}.

We recall briefly that a map $\Lambda$ is positive if and only if transform positive operators into positive operators  (a classical example of positive map is Transposition).

%

It is a relevant result that for low dimensional cases  every positive map $\Lambda$  assumes the following form \cite{S63,W76}:
\begin{equation}
\Lambda=\Lambda^{(1)}_{CP}+\Lambda^{(2)}_{CP}\circ T
\end{equation}
where $\Lambda^{(1)}_{CP},\Lambda^{(2)}_{CP}$ are CP maps, \emph{i.e.} maps that assume the Kraus-Stinespring form \cite{KBDW83}, and $T$ is transposition. Unfortunately this result cannot be extended for higher dimensional systems \cite{HHH96}.

The aim of this paper is to introduce a simple criterion to test if an observable is an EW. In particular, the paper is organized as follows.
In the first section  we give some definitions and lemmas used throughout  the paper. In  section \ref{sec.2}, basing on the simple case of a qubit, we  introduce the concept of \emph{separable tangent space}.
In sections \ref{sec.4} and \ref{sec.5} we present the main result of the paper, in particular the first one  is dedicated to prove the necessary part of the main theorem while the second regards the  sufficient condition.
Some concluding remarks are drawn in the final section.

\section{Some definitions}\label{sec.1}

In what follows we will indicate with $(\mathscr{H}^S,\langle\cdot|\cdot\rangle)$, or simply $\mathscr{H}^S$,
the Hilbert space associated to a quantum  system $S$.
It is natural to associate to an $n-$dimensional Hilbert space $\mathscr{H}$ the linear space $\mathcal{L}(\mathscr{H})$ of endomorphism on $\mathscr{H}$.
The space   $\mathcal{L}(\mathscr{H})$ is naturally endowed with the inner product:
\begin{equation}
(a,b)=\mathrm{Tr}(a^\dag b).
\end{equation}
It is simple to prove that the couple $(\mathcal{L}(\mathscr{H}),(\cdot,\cdot))$
is an $n^2-$dimensional complex Hilbert space.

A particularly interesting subset  of $\mathcal{L}(\mathscr{H})$ is
\begin{equation}
\mathcal{H}(\mathscr{H})=\{h\in\mathcal{L}(\mathscr{H}):h=h^\dag\}.
\end{equation}
This is the real linear space of Hermitian operators on $\mathscr{H}$. As it is well known, the set of Hermitian operators, in quantum mechanics, correspond to the set of all possible \emph{observables}.

A relevant subset of $\mathcal{H}(\mathscr{H})$ is:
\begin{equation}
\mathcal{H}^+(\mathscr{H})=\{a\in\mathcal{H}: \langle\psi|a|\psi\rangle\geq0, \forall\psi\in \mathscr{H}\}.
\end{equation}
This is the set of positive operators on $\mathscr{H}$. In particular, the set
$\mathcal{S}(\mathscr{H})$ of positive trace one operators will be called the \emph{set of states}.

Naturally all the previous definitions can be extended to the Hilbert space $\mathscr{H}^{\Sigma_{i=1}^k(i)}=\bigotimes_{i=1}^k\mathscr{H}^{(i)}$
associated to the  multipartite system
$S^{(\Sigma_{i=1}^k(i))}=\sum_{i=1}^k S^{(i)}$.

For typographical convenience, in what follows we put
\begin{equation}
\mathscr{H}^{(l\div m)}=\mathscr{H}^{\Sigma_{i=l}^m(i)}=\bigotimes_{i=l}^m\mathscr{H}^{(i)}
\end{equation}

We recall now the definitions of two important subsets of  $\mathcal{S}(\mathscr{H}^{(1\div k)})$:
\begin{definition}
We say \cite{W98} that a state $\varrho\in\mathcal{S}(\mathscr{H}^{(1\div k)})$
 is separable if and only if
\begin{equation}
\varrho=\sum_{i}p_{i}\bigotimes_{j=1}^k\varrho_i^{(j)},\,\,\,\,\,\,\,\,\,
\,\,\,  p_{i}\geq0,\,\, \Sigma_{i}p_{i}=1,
\end{equation}
where $\forall j=1,\ldots,k$ and $\forall i$, $\varrho^{(j)}_i\in \mathcal{S}({\mathscr H}^{(j)})$.
\end{definition}
The set of separable states will be denoted  by $\mathfrak{S}^{(1\div k)}$.

It is simple to note that $\mathfrak{S}^{(1\div k)}$ is the convex hull of the set $\mathfrak{P}^{(1\div k)}$ of pure product states:
\begin{align}
\mathfrak{P}^{(1\div k)}=&\left\{p=\bigotimes_{j=1}^k|e^{(j)}\rangle\langle e^{(j)}|\mbox{for some}\,\,|e^{(j)}\rangle\in{\mathscr H}^{(j)}, \right.\nonumber\\ &\left.\begin{array}{c}
\\
\\
\\
\end{array}\langle e^{(j)}|e^{(j)}\rangle=1,\forall j=1,\ldots,k                                                                      \right\}.
\end{align}

\begin{definition} We say \cite{W98} that a state $\varrho\in\mathcal{S}(\mathscr{H}^{(1\div k)})$
is entangled if and only if
\begin{equation}
    \varrho\in \mathfrak{E}(\mathscr{H}^{(1\div k)})= \mathcal{S}(\mathscr{H}^{(1\div k)})\setminus \mathfrak{S}(\mathscr{H}^{(1\div k)}).
\end{equation}
\end{definition}

In  this paper,we will focus our attention on the convex cone
generated by $\mathfrak{S}$. So we recall the following definition \cite{R96}:

\begin{definition}
A subset $C$ of a linear space $V$ is said a cone  if and only  if satisfies the
following condition:
\begin{equation}
x\in C  \Rightarrow px\in C, \forall p\geq 0.
\end{equation}
\end{definition}

Note that sometimes in the definition of a cone  the condition $\forall p\geq 0$ is replaced by the weaker $\forall p> 0$ \cite{R96}.

Important subspace of an Hilbert space $V$ are hyperplanes \cite{R96}:
 \begin{definition}
 Given an $n-$dimensional  Hilbert space $V$, the subspace $\mathcal{I}$ is said an  hyperplane if and only if $\mathrm{dim}\,\mathcal{I}=n-1$.
 \end{definition}

We recall briefly that an hyperplane $\mathcal{I}\subset V$ can be defined
by one of its, non zero, orthogonal vectors $a\in V$ \cite{R96}:
\begin{equation}
\mathcal{I}=\{x\in V|(x,a)=0\}.
\end{equation}
We will write $\mathcal{I}_a$ to denote the hyperplane defined by the vector $a$.

It is well known that  $\mathcal{I}_{a}$ defines the following  open half-spaces:
\begin{align}
\mathcal{I}^+_{a}&=\{x\in V: (x,a)>0\},\\
\mathcal{I}^-_{a}&=\{x\in V: (x,a)<0\},
\end{align}
and the corresponding closed half-spaces:
\begin{align}
\bar{\mathcal{I} }^+_{a}&=\{x\in V: (x,a)\geq0\},\\
\bar{ \mathcal{I}}^-_{a}&=\{x\in V: (x,a)\leq0\}.
\end{align}

We will use some simple results  obtained applying to hyper spaces and half-spaces some well known theorems of linear algebra. We will collect here these results to make the treatment more clear.

Let us consider $W\subseteq V$, a linear subspace of $V$. Let $\mathcal{I}_a$ be an hyper plane of $V$ determined by (one of) its orthogonal vector $a$. It is simple to prove, using  the projection theorem, that  the following lemmas hold:

\begin{lemma}\label{th.2.1}
The subspace $\mathcal{I}_a\cap W$ is determined by the following equality:
\begin{equation}
\mathcal{I}_a\cap W=\{x\in W: (x,a_W)=0\},
\end{equation}
where $a_W$ is the projection of $a$ on $W$.
\end{lemma}

\begin{lemma}\label{th.2.2}
The following equalities  hold:
\begin{align}
\mathcal{I}^+_{a}\cap W&=\{x\in W: (x,a_W)>0\},\\
\mathcal{I}^-_{a}\cap W&=\{x\in W: (x,a_W)<0\},
\end{align}
and
\begin{align}
\bar{\mathcal{I} }^+_{a}\cap W&=\{x\in W: (x,a_W)\geq0\},\\
\bar{ \mathcal{I}}^-_{a}\cap W&=\{x\in W: (x,a_W)\leq0\}.
\end{align}
\end{lemma}

As it was noted, an hyperplane $\mathcal{I}_a$  can be identified by  its  orthogonal monodimensional subspace (generated by $a$).
In particular, in the case  of $\mathcal{H}(\mathscr{H}^{(1\div k)})$, the vector $a$ is an observable that, under suitable conditions,  can detect entangled state:

\begin{definition}
If $(a,p)\geq0$ for all $p\in\mathfrak{S}^{(1\div k)}$ and at least one $q\in\mathfrak{E}^{(1\div k)}$  exists such that $(a,q)<0$ then $a$ is called an Entanglement Witness (EW).
\end{definition}

\section{Some geometrical considerations}\label{sec.2}

To expose the main  idea of the paper, let us start considering the state space $\mathscr{H}^2$ of  a single qubit. It is well known that, in the  Bloch sphere representation, states can be represented by point inside a  sphere in the $3-$dimensional real space generated by the Pauli matrices, $\sigma_1$, $\sigma_2$ and $\sigma_3$
where
\begin{equation}\label{eq.pm}
\sigma_1=\left(
           \begin{array}{cc}
             0 & 1 \\
             1 & 0 \\
           \end{array}
         \right)
,\quad\sigma_2=\left(
                \begin{array}{cc}
                  0 & -i \\
                  i & 0\\
                \end{array}
              \right)
,\quad\sigma_3=\left(
                 \begin{array}{cc}
                   1 & 0 \\
                   0 & -1 \\
                 \end{array}
               \right).
\end{equation}
 In fact it is simple to prove that every trace one Hermitian matrices acting on $\mathds{C}^2$ can be written as:
\begin{equation}\label{eqb}
h= \frac{1}{2}\mathds{I}+\frac{1}{2}(h_1\sigma_1+h_2\sigma_2+h_3\sigma_3)= \frac{1}{2}\mathds{I}+\frac{1}{2}\vec{h}\cdot\vec{\sigma}.
\end{equation}

It is possible to show that if $\|\vec{h}\| = \sqrt{ h_1^2 + h_2^2+ h_3^2}\leq 1$, (\ref{eqb}) is a positive matrix  and,  in particular, (\ref{eqb}) is a pure state if and only if $\|\vec{h}\| = 1$.

Let us find the tangent plane $\pi$ to the Bloch sphere through one  pure state, for example $p=\frac{1}{2}\mathds{I}+\frac{1}{2}\sigma_z$. It is immediate to note that the elements of $\pi$ has the following form:
\begin{equation} \label{eq3}
\sigma_3+\alpha\sigma_1+\beta\sigma_2,\quad \alpha,\beta\in\mathds{R},
\end{equation}
(see figure \ref{fig.1}).

\begin{figure}\label{fig.1}
\centering
 \psset{unit=1}
 \begin{pspicture}(-4,-2)(2,4)
\pstThreeDSphere[linecolor=black!60,SegmentColor={[cmyk]{0,0,0,0}}](0,0,0){2}
\pstThreeDLine[linewidth=1pt,linecolor=black,linestyle=dashed](0,0,0)(0,0,2)
\pstThreeDLine[linewidth=1pt,linecolor=black,linestyle=dashed](0,0,0)(0,2,0)
\pstThreeDLine[linewidth=1pt,linecolor=black,linestyle=dashed](0,0,0)(2,0,0)
{\psset{fillstyle=crosshatch,hatchwidth=0.1\pslinewidth,hatchsep=1\pslinewidth,hatchcolor=black!30}
\pstThreeDSquare(-2,-2,2)(4,0,0)(0,4,0)}
\pstThreeDCoor[xMin=2,xMax=4,yMin=2,yMax=4,
zMin=2,zMax=4,nameX=$\sigma_1$,nameY=$\sigma_2$,nameZ=$\sigma_3$,linecolor=black,linewidth=1pt]
\pstThreeDLine[linewidth=1pt,linecolor=black,arrows=->]%
(0,0,2)(1.5,0,2)
\pstThreeDLine[linewidth=1pt,linecolor=black,arrows=->]%
(0,0,2)(0,1.5,2)
\end{pspicture}\label{fig.1}
\end{figure}

The bidimensional subspace $C$ associated to the affine plane  (\ref{eq3}) is naturally $C=\mathrm{span}\{\sigma_1,\sigma_2\}$.

It is now immediate to note that  the hyperplane tangent to  $\mathcal{S}(\mathscr{H}^2)$  is the subspace:
\begin{equation}\label{eq.p1}
 \Pi=\mathrm{span}\{p,\sigma_1,\sigma_2\}.
\end{equation}

\begin{remark}
 By simple calculations, we see that the hyperplane $\Pi$ through the point $p$ can be obtained imposing that $\Pi$ contains the space  $C(p)$, where:
\begin{equation} \label{eqC}
C(p)=\{i[p,h]:h=h^\dag\}.
\end{equation}
\end{remark}
\begin{remark}
We note that the only vectors orthogonal to $ \pi$ must necessarily have the form
\begin{equation}
q=\alpha(\frac{1}{2}\mathds{I}-\frac{1}{2}\sigma_z).
\end{equation}
Note that $q$ is a positive operator if $\alpha\geq0$ while it is a negative operator if $\alpha\leq0$.
\end{remark}
In the next sections, we will prove that these remarks can be suitably generalized in the multipartite case.

\section{Main result (necessary condition)}\label{sec.4}

To simplify the notation we will introduce the following linear spaces
\begin{align}
\tau^{(i)}=\mathrm{span}&\{\mathds{I}^{(1\div(i-1))}\otimes h^{(i)}\otimes\mathds{I}^{((i+1)\div k)}:\\
&h^{(i)}\in\mathcal{H}(\mathscr{H}^{(i)}),\mathrm{Tr}h^{(i)}=0\}
\end{align}
where $\mathds{I}^{(i\div j)}$  is the identity operator on $\mathscr{H}^{(i\div j)}$ , obviously $\mathds{I}^{(i\div j)}=\bigotimes_{l=i}^j\mathds{I}^{(l)}$ where $\mathds{I}^{(l)}$  is the identity operator on $\mathscr{H}^{(l)}$.
Let us introduce the following linear space:
\begin{equation}
\tau=\mathrm{span}\left(\bigcup_{i=1}^k\tau^{(i)}\right).
\end{equation}
 Noticing that for all $t^{(i)}\in\tau^{(i)}$ and $t^{(j)}\in\tau^{(j)}$, $i\neq j$, $(t^{(i)},t^{(j)})=0$ we can write
 \begin{equation}
\tau =\bigoplus_{i=1}^k\tau^{(i)}.
\end{equation}
\begin{remark}
Note that $\tau$ is the set of the generators of local unitary transformations, i.e. the set of all unitary transformations $U$ on $\mathscr{H}^{(1\div k)}$ of the form
\begin{equation}
U^{(1\div k)}=\bigotimes_{i=i}^kU^{(i)},
\end{equation}
where $\forall i=1,\ldots,k$, $U^{(i)}$ is a unitary transformation on $\mathscr{H}^{(i)}$.
\end{remark}
The following set generalizes (\ref{eqC}).
\begin{definition}
Given an element $p\in\mathcal{H}(\mathscr{H}^{(1\div k)})$, we can define the following subspace:
\begin{equation}
\mathfrak{C}(p)=\{i[p,t]| t\in\tau\}.
\end{equation}
\end{definition}

If $p$ is a pure product states, we will call  $\mathfrak{C}(p)$ \emph{separable tangent space at the point $p$}.

Let us consider $k$ finite Hilbert space $\mathscr{H}^{(i)}$, $i=1,\ldots,k$ ($\mathrm{dim}\mathscr{H}^{(i)}=n$, $\forall i$).

Let us now  prove the following theorem:
\begin{theorem} \label{mth}
Let us consider  a vector $a\in\mathcal{H}(\mathscr{H}^{(1\div k)})$ and its corresponding hyperplane
$\mathcal{I}_a$. If  $a$ is an entanglement witness then  the following condition holds:
\begin{equation} \label{eq.nsc}
\mbox{if}\,\, p\in\mathfrak{P}^{(1\div k)}\cap\mathcal{I}_a \,\,\mbox{then}\,\, \mathfrak{C}(p)\subseteq\mathcal{I}_a.
\end{equation}
\end{theorem}

\begin{proof}
Let us consider  a vector $a\in\mathcal{H}(\mathscr{H}^{(1\div k)})$ and its corresponding hyper plane
$\mathcal{I}_a$.

If $\mathfrak{P}\cap\mathcal{I}_a=\emptyset$ we have nothing to prove, so let us suppose that $\mathfrak{P}\cap\mathcal{I}_a\neq\emptyset$ i.e., a  pure product state $p$ exists  that lies in $\mathcal{I}_a$.
In particular let us put:
\begin{equation}
p=\bigotimes_{i=1}^k|e_1^{(i)}\rangle\langle e_1^{(i)}|,
\end{equation}
where $\forall i=1,\ldots,k$ $|e_1^{(i)}\rangle\in \mathscr{H}^{(i)}$. Naturally we can choose  orthonormal bases $\mathcal{B}^{(i)}$ in $\mathscr{H}^{(i)}$  such that $|e_1^{(i)}\rangle \in\mathcal{B}^{(i)}$, in particular let us put $\mathcal{B}^{(i)}=\{|e_j^{(i)}\rangle\}_{j=1,\ldots,n}$.
Using this representation, $p$ is represented by an $\underbrace{n^2\times \cdots\times n^2}_{k-\mbox{times}}$ hermitian matrix:
\begin{equation}
p= \bigotimes_{i=1}^k p_1^{(i)},
\end{equation}
where $p^{(i)}_{lm}=\delta_{1l}\delta_{1m}$.

Let us first introduce the following trace less hermitian matrices: $\{\sigma_1(lm),\sigma_2(lm),\sigma_3(lm)\}_{l<m=1,\ldots,n}$
whose non zero elements are so defined:
\begin{align}
({\sigma_1(lm)})_{ij}&=\delta_{li}\delta_{mj}+\delta_{lj}\delta_{mi},\\
{(\sigma_2(lm))}_{ij}&=i\delta_{li}\delta_{mj}-i\delta_{lj}\delta_{mi},\\
{(\sigma_3(lm))}_{ij}&=\delta_{li}\delta_{lj}-\delta_{mj}\delta_{mi}.
\end{align}

The notation emphases the fact that the matrices introduced generalize the Pauli matrices (\ref{eq.pm}).

Note that
\begin{equation}
\{\sigma_1(lm),\sigma_2(lm),\sigma_3(lm)\}_{l<m=1,\ldots,n},
\end{equation}
 is a spanning set for the set of all Hermitian trace less matrices.

Using this representation, we note that
\begin{equation}
\bigotimes_{i=1}^{j-1}|e_1^{(i)}\rangle\langle e_1^{(i)}|\otimes{\sigma_1(1m)}\otimes\bigotimes_{i=j+1}^{k}|e_1^{(i)}\rangle\langle e_1^{(i)}|\in\mathfrak{C}(p),\label{eq.42}
\end{equation}
and
\begin{equation}
 \bigotimes_{i=1}^{j-1}|e_1^{(i)}\rangle\langle e_1^{(i)}|\otimes{\sigma_2(1m)}\otimes\bigotimes_{i=j+1}^{k}|e_1^{(i)}\rangle\langle e_1^{(i)}|\in\mathfrak{C}(p)\label{eq.43}.
\end{equation}
 for all $m=2,\ldots,n$ and $j=1,\ldots,k-1$.
Moreover, it is trivial to observe that  the previous operators are a spanning set for $\mathfrak{C}(p)$.
Let us now prove that $\mathfrak{C}(p)$ is a subspace of $\mathcal{I}_a$, in particular it is sufficient  to prove that the elements (\ref{eq.42}) and (\ref{eq.43}) belong to $\mathcal{I}_a$. We start proving  that ${\sigma_1(12)}\otimes\bigotimes_{i=2}^{k}|e_1^{(i)}\rangle\langle e_1^{(i)}|\in \mathcal{I}_a$ and ${\sigma_2(12)}\otimes\bigotimes_{i=2}^{k}|e_1^{(i)}\rangle\langle e_1^{(i)}|\in \mathcal{I}_a$.

Let us consider  the subspaces
\begin{align}
W_{12}^{(1)}&=\mathrm{span}\left\{|e_1^{(1)}\rangle,|e_2^{(1)}\rangle\right\},\\
\mathcal{W}_{12}^{(1)}&=\mathrm{span}\left\{|e_1^{(1)}\rangle\otimes\bigotimes_{i=2}^{k}
|e_1^{(i)}\rangle,|e_2^{(1)}\rangle\otimes\bigotimes_{i=2}^{k}|e_1^{(i)}\rangle \right\},
\end{align}
and the corresponding spaces $\mathcal{H}({W}_{12}^{(1)})$ and $\mathcal{H}(\mathcal{W}_{12}^{(1)})$.
Using lemmas \ref{th.2.1} and \ref{th.2.2}
applied to the space $\mathcal{W}_{12}^{(1)}$, we obtain that
\begin{equation}
\mathcal{W}_{12}^{(1)}\cap \bar{\mathcal{I}}_a^+=\{x\in \mathcal{W}_{12}^{(1)}: (x,a_\mathcal{W})\geq0\},
\end{equation}
where $a_\mathcal{W}$ is the projection of $a$ on the subspace $\mathcal{H}(\mathcal{W}_{12}^{(1)})$.
Now we recall that $\mathfrak{S}^{(1\div k)}\subset  \bar{\mathcal{I}}_a^+$. Immediately we have that
$\mathcal{H}(\mathcal{W}_{12}^{(1)})\cap\mathfrak{S}^{(1\div k)}\subset  \mathcal{H}(\mathcal{W}_{12}^{(1)})\cap\bar{\mathcal{I}}_a^+$.

So we have that
$(x,a_\mathcal{W})\geq 0, \forall x\in\mathcal{H}(\mathcal{W}_{12}^{(1)})\cap\mathfrak{S}^{(1\div k)}$.
It is simple to prove that
\begin{equation}
\mathcal{H}(\mathcal{W}_{12}^{(1)})\cap\mathfrak{S}^{(1\div k)}=	\left\{p^{(1)}\otimes\bigotimes_{i=2}^{k}|e_1^{(i)}\rangle\langle e_1^{(i)}| , p^{(1)}\in\mathcal{S}(W_{12}^{(1)})\right\}.
\end{equation}

This simple remark implies that $a_\mathcal{W}$ must assume the form
$|e_2^{(1)}\rangle\langle e_2^{(1)}|\otimes\bigotimes_{i=2}^{k}|e_1^{(i)}\rangle\langle e_1^{(i)}|$ or $0$.
In both cases it is simple to prove that

\begin{align}
&\left(\tilde{\sigma}_1{(12)},a_\mathcal{W}\right)=0,\\
&\left(\tilde{\sigma}_2{(12)},a_\mathcal{W}\right)=0,
\end{align}
where
\begin{align}
\tilde{\sigma}_1{(12)}=&\sigma_1{(12)}\otimes\bigotimes_{i=2}^{k}|e_1^{(i)}\rangle\langle e_1^{(i)}|,\\
\tilde{\sigma}_2{(12)}=&\sigma_2{(12)}\otimes\bigotimes_{i=2}^{k}|e_1^{(i)}\rangle\langle e_1^{(i)}|.
\end{align}
Repeating the same reasoning for the different space
$\mathcal{W}_{1m}^{(l)}$, $l=2,\ldots, k$ and $m=2,\ldots,n$,
we can complete the proof.
\end{proof}

Theorem \ref{mth} has an immediate geometrical interpretation. It is evident, following the proof of the theorem, that equation (\ref{eq.nsc}) states that if an Entanglement Witness $a$ is othogonal to some pure product  state
$p=\bigotimes_{i=1}^k|e_1^{(i)}\rangle\langle e_1^{(i)}|\in \mathcal{I}_a$ then $\mathcal{I}_a$ is tangent to all possible Bloch Sphere $\mathcal{S}(\mathcal{W}^{(l)}_{12})$
where
\begin{equation}
 \mathcal{W}^{(l)}_{12}=\mathrm{span}\left\{|\psi_1^l\rangle,|\psi_2^l\rangle|\right\}
\end{equation}
with
\begin{align}
|\psi_1^l\rangle&=\bigotimes_{j=1}^{l-1}|e^{(j)}_1\rangle\otimes|e^{(l)}_1\rangle\otimes\bigotimes_{j=l+1}^k|e^{(j)}_1\rangle\\
|\psi_2^l\rangle&=\bigotimes_{j=1}^{l-1}|e^{(j)}_1\rangle\otimes|e^{(l)}_2\rangle\otimes\bigotimes_{j=l+1}^k|e^{(j)}_1\rangle
\end{align}
and
$|e^{(l)}_1\rangle\langle e^{(l)}_1|\neq|e^{(l)}_2\rangle\langle e^{(l)}_2|$.

\section{Main result (sufficient condition)}\label{sec.5}

In the previous section we proved that if a hyperplane, which defines an entanglement witness,  contains a pure product states $p$, then necessarily must contain the separable tangent spaces $\mathfrak{C}(p)$.

In this section we will prove that this condition is not only necessary, but also sufficient.

\begin{theorem}\label{thsc}
Let us consider  a vector $a\in\mathcal{H}(\mathscr{H}^{(1\div k)})$   and its corresponding hyperplane $\mathcal{I}_a$.
If  $\mathcal{I}_a$ satisfies condition (\ref{eq.nsc})
then $a$  satisfies the following equation:
\begin{equation}\label{eq.nsc2}
 (p_1,a)(p_2,a)\geq0
\end{equation}\label{eq.sp}
for all $p_1,p_2\in\mathfrak{P}^{(1\div k)}$.
\end{theorem}

Note that equation (\ref{eq.sp}) implies that if exists a pure product state $p_1(p_2)\in\mathfrak{P}^{(1\div k)}$ such that $(p_1,a)<0$ ($(p_2,a)>0$) then
$(p,a)\leq0$ (respectively $(p,a)\geq0$) for all $p\in\mathfrak{P}^{(1\div k)}$.

\begin{proof}

We will show inductively on the number of subsystems $k$ that theorem \ref{thsc} holds.

In particular we will  prove that
\begin{enumerate}
  \item theorem \ref{thsc} holds for $k=1$;
  \item if theorem \ref{thsc} holds for  $k-1$ then it holds for $k$.
\end{enumerate}

\textbf{Step 1} ($k=1$).
Note that in this case, obviously, $\mathfrak{P}=\mathcal{S}(\mathscr{H})$.

Let us suppose that  $\mathcal{I}_a$ satisfies condition (\ref{eq.nsc})
but  not condition (\ref{eq.nsc2}).
This means that, at least, a couple of  pure states, for example
$p_1=|e_1\rangle\langle e_1|$ and $p_2= |e_2\rangle\langle e_2|$, exist such that
\begin{equation}\label{eq.50}
   (p_1,a)<0,\quad  (p_2,a)>0.
\end{equation}
Without loss of generality, see lemmas \ref{th.2.1} and \ref{th.2.2}, we can restrict our attention to the bidimensional space $W$ generated by $|e_1\rangle$ and $|e_2\rangle$.
It is evident that $a_\mathcal{W}$, i.e. the projection of $a$ on the space $\mathcal{W}=\mathcal{H}(W)$, cannot be a definite operator, so we can choose a representation such that
\begin{equation}
a_\mathcal{W}=\alpha\frac{1}{2} (\mathds{I}+\beta\sigma_3),
\end{equation}
where $\alpha,\beta\in\mathds{R}$ and $|\beta|>1$, $\alpha\neq0$.

It is simple to show that
\begin{align}
p_3&=\frac{1}{2} \mathds{I}+\frac{1}{2}\delta \sigma_1-\frac{1}{2\beta}\sigma_3\in\mathcal{I}_a\cap\mathfrak{P},\\
p_4&=\frac{1}{2} \mathds{I}+\frac{1}{2}\delta \sigma_2-\frac{1}{2\beta}\sigma_3\in\mathcal{I}_a\cap\mathfrak{P},
\end{align}
where $\delta=\sqrt{1-\beta^{-2}}$.

Using the hypothesis (\ref{eq.nsc}),
we obtain immediately that $\{\mathds{I}, \sigma_1, \sigma_2, \sigma_3\}\subset\mathcal{I}_a $. It is evident that this last implies
\begin{equation}
   (p_1,a)=0\: \mbox{and} \: (p_2,a)=0,
\end{equation}
against the hypothesis (\ref{eq.50}), and so theorem \ref{thsc} holds for $k=1$.
\begin{remark}
Note that for $k=1$, if exists a pure states $p\in \mathcal{S}(\mathscr{H})$ such that  $(p,a)>0$ then   equation (\ref{eq.nsc}) represents a necessary and sufficient positivity condition for $a$. This supports the common idea that Entanglement Witness generalizes the concept  of positive operator.
\end{remark}

\begin{remark}\label{rem.5.2}
By direct calculations, it is possible to show, in the case of $k=1$,  that if $\mathcal{I}_a$ satisfies (\ref{eq.nsc}) and  contains the projectors associated to two  linearly  independent vectors $|e_1\rangle,|e_2\rangle\in\mathscr{H}$,  then, for all   linear combinations $|g\rangle=\alpha_1|e_1\rangle+\alpha_2|e_2\rangle$, $\alpha_1,\alpha_2\in \mathds{C}$,  $\mathcal{I}_a$ contains the projectors
\begin{equation}
g=|g\rangle\langle g|.
\end{equation}
Obviously, this results can be generalized to multipartite systems in the following way.
If  $\mathcal{I}_a$ satisfies (\ref{eq.nsc}) and  contains the projectors associated to the vectors $|e_1^{(1)}\rangle\otimes|f\rangle$ and $|e_2^{(1)}\rangle\otimes|f\rangle$, where $|f\rangle$ is a product vector in $\mathscr{H}^{(2\div k)}$, and
  $|e_1\rangle,|e_2\rangle\in\mathscr{H}^{(1)}$ are linearly independent,  then, for all   linear combinations $|g\rangle=\alpha_1|e_1\rangle+\alpha_2|e_2\rangle$, $\alpha_1,\alpha_2\in \mathds{C}$,  $\mathcal{I}_a$ contains the projectors
\begin{equation}
g'=|g\rangle\langle g|\otimes|f\rangle\langle f|.
\end{equation}
\end{remark}

\textbf{Step 2.}

Now suppose that  theorem \ref{thsc} holds for $k-1$ component systems. We will show now that it holds for $k$ component systems.
In fact let us now suppose that an hyperplane $\mathcal{I}_a$ exists  such that
\begin{equation}
\forall \,\, p\in\mathfrak{P}^{(1\div k)}\cap\mathcal{I}_a \,\,\Rightarrow\,\, \mathfrak{C}(p)\subseteq\mathcal{I}_a
\end{equation}
but $a$ does not satisfie condition (\ref{eq.nsc2}).

This means that at least two pure product states, for example  $p_1$ and $p_2$, exist such that
\begin{align}
p_1=&\bigotimes_{i=1}^{k}|e_1^{(i)}\rangle\langle e_1^{(i)}|\\
p_2=&\bigotimes_{i=1}^{k}|e_2^{(i)}\rangle\langle e_2^{(i)}|
\end{align}
such that
\begin{equation} \label{h1}
(a,p_1)<0;\,\,\,(a,p_2)>0.
\end{equation}

Let us consider the spaces
\begin{equation}
W_{12}^{(i)}=\mathrm{span}\{|e_1^{(i)}\rangle,|e_2^{(i)}\rangle\}.
\end{equation}

Now two cases arise:

Case 1)
\begin{equation} \label{eq.1}
\mathrm{dim}W_{12}^{(i)}=1,
\end{equation}
 for at least one $i$;

Case 2)
\begin{equation} \label{eq.2}
\mathrm{dim}W_{12}^{(i)}=2
\end{equation}
 for all $i$.

In the first case, obviously, the situation can be reduced to the $k-1$ component systems case.
In fact, by equation (\ref{eq.1}) at least one $i$ exists, for example $i=1$, such that  $|e_1^{(1)}\rangle=|e_2^{(1)}\rangle$. In this case, let us consider the space
\begin{equation}
W^{(1\div k)}=\bigotimes_{i=1}^k W_{12}^{(i)}.
\end{equation}
The corresponding space $\mathcal{H}(W^{(1\div k)})$ contains only elements $h$ of the form
\begin{equation}
h=|e_1^{(1)}\rangle\langle e_1^{(1)}|\otimes h',
\end{equation}
where $h'\in \mathcal{H}(W^{(2\div k)})$, whit $W^{(2\div k)}=\bigotimes_{i=2}^k W_{12}^{(i)}$.
In particular, this simple remark  implies
that, given two elements $a_1,a_2\in \mathcal{H}(W^{(1\div k)})$,  we have
\begin{align}
a_1&=|e_1^{(1)}\rangle\langle e_1^{(1)}|\otimes a_1',\\
a_2&=|e_1^{(1)}\rangle\langle e_1^{(1)}|\otimes a_2',
\end{align}
and
\begin{equation}
(a_1,a_2)=(a_1',a_2'),
\end{equation}
where   the inner product on the lhs (rhs) of the previous equation must be intended in the  Hilbert space $\mathcal{H}(\mathscr{H}^{1\div k})$ (respectively, $\mathcal{H}(\mathscr{H}^{2\div k})$).
This last gives
\begin{equation} \label{h1}
(a'_W,p'_1)<0;\,\,\,(a'_W,p'_2)>0,
\end{equation}
where
 \begin{align}
p'_1=&\bigotimes_{i=2}^{k}|e_1^{(i)}\rangle\langle e_1^{(i)}|,\\
p'_2=&\bigotimes_{i=2}^{k}|e_2^{(i)}\rangle\langle e_2^{(i)}|,
\end{align}
and
\begin{equation}
   a_W=|e_1^{(1)}\rangle\langle e_1^{(1)}|\otimes a'_W.
\end{equation}
($a_W$ is the projection of $a$ on the space $\mathcal{H}(W^{(1\div k)})$).
Let us consider the $k-$partite system $W^{(2\div k)}$.
It is simple to verify that $\mathcal{I}_{a'_W}(\subseteq\mathcal{H}(W^{(2\div k)})$ satisfies the following condition:
\begin{equation}
\mbox{if}\,\, p\in\mathfrak{P}^{(2\div k)}\cap\mathcal{I}_{a'_W}\,\,\mbox{then}\,\, \mathfrak{C}(p)\cap \mathcal{H}(W^{(2\div k)}) \subseteq\mathcal{I}_{a'_W}.
\end{equation}
We note  that:
\begin{align}
&\mathfrak{P}^{(2\div k)}\cap\mathcal{I}_{a'_W}\subseteq\mathfrak{P}^{(2\div k)}\cap\mathcal{H}(W^{(2\div k)})=\\
&=\left\{p=\bigotimes_{j=2}^k|e^{(j)}\rangle\langle e^{(j)}| \mbox{for some}\,\,|e^{(j)}\rangle\in{W^{(j)}_{12}},\forall j=2,\ldots,k                                                                      \right\},
\end{align}
and
\begin{equation}
 \mathfrak{C}(p)\cap \mathcal{H}(W^{(2\div k)})=\{i[p,t]| t\in\tau^{(2\div k)}\cap\mathcal{H}(W^{(2\div k)})\}.
\end{equation}
Using  the inductive hypothesis, we obtain  immediately a contradiction.\\

Case 2).
 By equation (\ref{eq.2}) we have $|e_1^{(i)}\rangle\neq|e_2^{(i)}\rangle$ for all $i$.
In this case, the space $W^{(1\div k)}$ is a $2^k$ dimensional Hilbert space. Let us consider the following
two pure states in $\mathcal{H}(W)$:
\begin{align}
p_3&=|e_1^{(1)}\rangle\langle e_1^{(1)}|\otimes\bigotimes_{i=2}^k|e_2^{(i)}\rangle\langle e_2^{(i)}|,\\
p_4&=|e_2^{(1)}\rangle\langle e_2^{(1)}|\otimes\bigotimes_{i=2}^k|e_1^{(i)}\rangle\langle e_1^{(i)}|.
\end{align}

It is simple to show that
\begin{equation}
(p_3,a_W)=0\:\mbox{and}\: \quad (p_4,a_W)=0.
\end{equation}
In fact if $(p_3,a_W)>0$ $((p_3,a_W)<0)$, by a similar arguments of case 1) applied to the states $p_1$ and $p_3$ (respectively $p_2$ and $p_3$) we have immediately a contradiction. Similar reasoning holds for $p_4$.

Let us now consider the state $|g^{(1)} \rangle\langle g^{(1)}|$ where $|g^{(1)}\rangle \in\mathscr{H}^{(1)}$ is a non trivial linear combination of $|e_1^{(1)}\rangle$ and $|e_2^{(1)}\rangle$:
\begin{equation}
|g^{(1)} \rangle=\alpha_1|e_1^{(1)}\rangle+\alpha_2|e_2^{(1)}\rangle,
\end{equation}
 $\alpha_1,\alpha_2\neq0$, and the following  multipartite states
\begin{align}
p_5=&|g^{(1)} \rangle\langle g^{(1)}|\otimes\bigotimes_{i=2}^k|e_1^{(i)}\rangle\langle e_1^{(i)}|,\\
p_6=&|g^{(1)} \rangle\langle g^{(1)}|\otimes\bigotimes_{i=2}^k|e_2^{(i)}\rangle\langle e_2^{(i)}|.
\end{align}
By similar arguments of  case 1, we can exclude $(a,p_5)>0$ and $(a,p_6)<0$.
Now let us suppose that $(a,p_5)<0$ and $(a,p_6)>0$ hold simultaneously. We have an immediate contradiction using the same arguments of case 1.
So at least one of the following possibilities must hold:
\begin{enumerate}
  \item $(a,p_5)=0$;
  \item $(a,p_6)=0$.
\end{enumerate}
The first possibility, using the fact that  $(a,p_4)=0$ and the hypothesis (\ref{eq.nsc}), implies that $(a,p_1)=0$ (see remark \ref{rem.5.2}) against the hypothesis (\ref{eq.50}).
The second  case, using the fact that  $(a,p_3)=0$ and the hypothesis (\ref{eq.nsc}), implies that $(a,p_2)=0$ against the hypothesis (\ref{eq.50}).
So we have a contradiction and then we can conclude that theorem \ref{thsc} holds for $k$ component systems. This completes the proof.
\end{proof}

\section{Concluding remarks}\label{sec.6}

In the previous sections we give the necessary tools to formulate a necessary and sufficient condition for Entanglement Witness. In particular, using theorem \ref{mth} and  theorem \ref{thsc} we obtain immediately the following theorem.
\begin{theorem}\label{thfi}
Let us consider  a vector $a\in\mathcal{H}(\mathscr{H}^{(1\div k)})$   and its corresponding hyperplane $\mathcal{I}_a$.
Then $a$ is an Entanglement Witness if and only if the following conditions hold:
\begin{enumerate}
 \item $\exists p\in\mathfrak{P}^{(1\div k)}: (p,a)>0$;
\item $\mbox{if}\,\, p\in\mathfrak{P}^{(1\div k)}\cap\mathcal{I}_a \,\,\mbox{then}\,\, \mathfrak{C}(p)\subseteq\mathcal{I}_a$;
\item  $a\notin\mathcal{H}^+(\mathscr{H}^{(1\div k)})$.
\end{enumerate}
\end{theorem}

Note that conditions 1 and 2 are necessary and sufficient for $a$ to be positive on the set of  separable states, while condition 3 assures that at least an entangled state  exists that is detected by $a$.

Moreover, as it was noted in the introduction, this theorem can be translated into a criterion for positive maps via the CJ-isomorphism.
Let us consider in fact, the following entangled bipartite vector:
\begin{equation}
|\alpha \rangle=\sum_{i=1}^n|e_i^{(1)} \rangle|e_i^{(2)} \rangle
\end{equation}
where $|e_i^{(1)} \rangle\in\mathscr{H}^{(1)}$, $|e_i^{(2)} \rangle\in\mathscr{H}^{(2)}$ and $\langle e_j^{(l)} |e_i^{(l)} \rangle=\delta_{ij}$, ($l=1,2$).
Given  a map $\Lambda:\mathcal{L}(\mathscr{H}^{(1)})\rightarrow \mathcal{L}(\mathscr{H}^{(1)})$, from the CJ-isomorphism, it is simple to show that
\begin{theorem}

The map $\Lambda$ is a positive, but not completely positive, map if and only if
$
\Lambda\otimes \mathbf{I}^{(2)}(|\alpha \rangle\langle \alpha|)
$, where $\mathbf{I}^{(2)}$ is the identity map on $\mathcal{L}(\mathscr{H}^{(2)})$,
is an entanglement witness on $\mathcal{H}(\mathscr{H}^{(1)}\otimes\mathscr{H}^{(2)})$.
\end{theorem}

We conclude  observing that we can improve the criterion introduced in this paper. In fact, we can limit  our attention in theorems \ref{mth} and \ref{thsc} only on a linear independent set of pure product states in $\mathfrak{P}^{(1\div k)}\cap\mathcal{I}_a$.

This last is an immediate consequence of the following trivial remark:
\begin{align}
\mbox{if}& \,p=\sum_{i=1}^l\alpha_i p_i, \alpha_i\in \mathds{R}, p_i\in\mathfrak{P}^{(1\div k)}\cap\mathcal{I}_a \\
&\mbox{then}\,\mathfrak{C}(p)\subseteq \mathfrak{C}(p_1)+\ldots+\mathfrak{C}(p_l).
\end{align}

Finally we observe that the set $\mathfrak{P}^{(1\div k)}\cap\mathcal{I}_a$  was widely studied in \cite{LKCH00} in connection to the optimization of EW. There, the authors gave some useful method to determine this set and indeed theorem \ref{thfi} gives a practical method to check if an observable is an EW.

\bibliography{bibliografia}

\end{document}